\documentclass[12pt,a4paper]{article}
\usepackage[english]{babel}
\usepackage{amsmath,amssymb,amsthm,epsfig,dsfont,color}

\newtheorem{theo}{Theorem}[]
\newtheorem{cor}[theo]{Corollary}
\newtheorem{lem}[theo]{Lemma}
\newtheorem{prop}[theo]{Proposition}
\newtheorem{defn}[theo]{Definition}
\newtheorem{remark}[theo]{Remark}

\newcommand{\R}{\mathds{R}}

\newcommand{\N}{\mathds{N}}
\newcommand{\I}{\mathds{I}}

\newcommand{\keywords}[1]{\par\addvspace\baselineskip\noindent\textbf{Keywords:}\enspace\ignorespaces#1}
\newcommand{\AMSclassification}[1]{\par\addvspace\baselineskip\noindent\textbf{Mathematical subject classification:}\enspace\ignorespaces#1}
\newcommand{\acknowledgement}[1]{\par\addvspace\baselineskip\noindent\textbf{\small{Acknowledgement.}}\enspace\ignorespaces\small#1}

\title{Average sex ratio and population maintenance cost}
\author{
\small{Eduardo Garibaldi}\\
\footnotesize{UNICAMP -- Departamento de Matem\'atica}\\
\footnotesize{13083-859 Campinas - SP, Brasil}\\
\footnotesize{\texttt{garibaldi@ime.unicamp.br}}
\and
\small{Marcelo Sobottka}\\
\footnotesize{UFSC -- Departamento de Matem\'atica}\\
\footnotesize{88040-900 Florian\'opolis - SC, Brasil}\\
\footnotesize{\texttt{sobottka@mtm.ufsc.br}}
}
\date{}

\begin{document}

\maketitle

\begin{abstract}
The ratio of males to females in a population is a meaningful characteristic of sexual species. The reason
for this biological property to be available to the observers of nature seems to be a question never asked.
Introducing the notion of historically adapted populations as global minimizers of maintenance cost functions,
we propose a theoretical explanation for the reported stability of this feature. This mathematical formulation
suggests that sex ratio could be considered as an indirect result shaped by the antagonism between the size of the population and the finiteness of resources.


\AMSclassification{37N25, 92D15, 92D25.}
\end{abstract}

\bigskip
\hrule
\noindent
{\footnotesize\em This is a pre-copy-editing, author-produced PDF of an article accepted for publication in SIAM Journal on Applied Mathematics, following peer review. The definitive publisher-authenticated version {\em E. Garibaldi and M. Sobottka. Average sex ratio and population maintenance cost. SIAM Journal on Applied Mathematics (2011), 71, 1009-1025, doi:10.1137/100817310 }, is available online at: http://epubs.siam.org/doi/10.1137/100817310 .}
\hrule
\bigskip

\section{Introduction}

\begin{quotation}
\sffamily{\small{
``I formerly thought that when a tendency to produce the two sexes in equal numbers was advantageous to the species,
it would follow from natural selection, but I now see that the whole problem is so intricate that it is safer to leave
its solution for the future.''
}}

\vspace{.2cm}

\sffamily{\small{
Charles Darwin in \emph{The descent of man} \cite{Darwin}.
}}
\end{quotation}

It was reported \cite{Vollrath} about five females for each male in Panamaniam colonies of spider \emph{Anelosimus
eximius}\sloppy. This arachnid species has developed social groupings with overlap of generations, cooperation in care
of young and reproductive division of labour. Male-biased sex ratio has been observed \cite{ABGCT} in French populations
of marmot \emph{Marmota marmota}. This socially monogamous mammal is a cooperative breeding species and subordinate
males participate in social thermoregulation during winter. In a territorial bird species, the Seychelles warbler
\emph{Acrocephalus seychellensis}, it was identified \cite{Komdeur} a facultative adjustment of offspring sex ratio.
The role of daughters as helpers in raising subsequent broods and the quality of a territory classified according to
its size, the density of insect prey and the amount of foliage are factors that explain the sex ratio shift
in offsprings, from mainly females on high-quality territories to mainly males on low-quality territories.

Sex-ratio studies form a fascinating topic in evolutionary biology, which underline the impact of natural selection on
the allocation of resources to male and female progeny. Using a frequency-dependent argument, Fisher provided
\cite{Fisher} a theoretical explanation for the prevalence of near 1:1 sex ratio under natural selection.
The effort to understand the stability of biased sex ratios has enabled the central theory to find successive
and fruitful extensions. For instance, Hamilton's local mate competition hypothesis \cite{Hamilton} was
originally introduced to clarify how the interactions between siblings produce very female-biased sex ratios
in parasitic wasps.

It has been very useful in sex-ratio theory the point of view which consists in describing collective phenomena from
the actions and expectations of individuals. Charnov's book on sex allocation \cite{Charnov} is an extremely successful
illustration of this tendency. By considering non-linear and unequal returns from parental investment in sons and
daughters, Charnov has developed a nice mathematical formulation, able not only to conceptually explain cases of both
Fisherian and non-Fisherian sex ratios, but also to provide predictions to be tested in experiments. Another
example of a fundamental contribution from the philosophical approach based on methodological individualism is the
so-called Triver-Willard hypothesis \cite{TW}, which suggests that natural selection should favor parental ability
to adjust the sex ratio of their offspring in response to environmental conditions.

The focus on individual behavior leads to the important discussion about selection criteria for reproductive strategies.
Nowadays questions arising from parent-offspring conflict, parental investment, sibling antagonism and mate choice may
be mathematically addressed by evolutionary game dynamics (see, for instance, \cite{HK}).

We will adopt a different point of view, which consists mostly in a global perspective
by proposing a population-based optimization model. As any general model, this mathematical formulation will have
mainly heuristic purposes. Focusing on the entire population as a dynamical agent without directly paying attention to
specific biological parameters, the consideration of an implicit maintenance cost function will give qualitative
insights for a common biological feature: an observable sex ratio. As a matter of fact, our main result will argue in
favor of the hypothesis that the very possibility of a sex ratio being recognized in the nature may reflect a balance
between the size of the population and the finiteness of resources.

The mathematical techniques developed here have foundations in common with the variational
theory applied to the study of ground-states of generalized Frenkel-Kontorova models on a one-dimensional crystal
(see, for example, \cite{ALeD,Bangert}). Actually, statistical physics methods
have been already successfully exploited in evolutionary games on graphs, specially when
social networks are seen as the result of individual interactions governed by some kind of
interdependency, such as sexual relationships (see, for example, \cite{SF}).

In order to be more concrete, suppose we periodically census the size of each gender in some biological population.
Let $\binom{x_i}{y_i}\in\N^2$ be the $(i + 1)^{th}$ census, where $x_i$ and $y_i$ indicate the number of females and males, respectively.
An infinite list $\omega=(\binom{x_0}{y_0}, \binom{x_1}{y_1}, \binom{x_2}{y_2}, \ldots )$ can be
viewed as a possible (yet maybe unlikely) historical record of each gender of a particular population.
Obviously, not all $\omega $ has a biological meaning: this could be the case, for instance, of
$\omega=(\binom{1}{0}, \binom{0}{0}, \binom{0}{0}, \binom{0}{99}, \binom{99}{0},\ldots)$.
Hence one evidently needs some criteria to select among all registers
those which may indeed represent a possible history of some population. This can be obtained by considering a function
which associates some maintenance cost for any finite register of a population history. Such a cost function shall
necessarily capture chief features\footnote{In particular, it must assign a high cost to finite population histories
which should be unlikely.} of the biological population to be modeled.

A cost function leads us to the notion of historically adapted populations, which intuitively
correspond to those populations more efficient in the use of available resources. The concept
of ``historically adapted population'' shall not be understood as ``survival of the fittest''.
In fact, we are not focusing on competition either between species or among individuals, but only
looking for the optimal rates for each gender in populations under certain environmental conditions.
In particular, we neither claim that actual populations are historically adapted nor try to explain biological mechanisms which could lead a population to be historically adapted. Even so,
the mathematical proof of a kind of
abundance of historically adapted populations with an identifiable sex ratio might insinuate why sex-ratio random variations in a given population
seem to be a very rare phenomenon in nature.

The paper is organized as follows. In section~\ref{basics}, we present the mathematical
model we shall study. In section~\ref{ExistenceHistoricalyAdaptedPop}, we introduce the notion of historically adapted
populations and show some of their properties. In section~\ref{AverageSexRatioSection}, we present arguments for the
existence of an asymptotic average sex ratio for historically adapted populations. Concluding remarks are discussed
in section~\ref{Discussion}. In appendix A, one can find the mathematical proofs of the results used
along the paper.

\section{The model}\label{basics}

In this section, we shall present a mathematical formulation to model two-sex populations
in a finite-resource environment. Denote then the set of all nonnegative integers by $\N$. Define
\begin{equation*}\label{StateSpace} \Omega:=\left(\N^2\right)^\N:=\left\{\binom{x_i}{y_i}_{i\in\N}:\ x_i, y_i \in \N,\ \forall \, i\in\N\right\}.\end{equation*}
The elements of $\Omega$ will be called the (possible) histories for the population. Each history $\omega\in\Omega$ can be interpreted as a list of consecutive censuses of female and male populations. Given $\omega=\binom{x_i}{y_i}_{i\in\N}\in\Omega$ and $m,n\in\N$ with $m\leq n$, we set $\omega[m]:=\binom{x_m}{y_m}$ and $\omega{[m,n]}:=\binom{x_i}{y_i}_{m\leq i\leq n}$, which are the restrictions of the infinite history $\omega$ to the moment $m$ and to the finite history from the moment $m$ until the moment $n$, respectively.

We will now define a class of cost functions which shall be used to select those censuses that may
in fact be realized. First of all, we would like to emphasize that, although we will explicitly express only the
dependence on population sizes, a maintenance cost function must depend on several biological and physical variables. We will just omit this multiple dependence in our analysis. Thus, let $C:\N^2\times\N^2\to \R$ be a function bounded from below, which means
\begin{equation}\label{BoundedBelow}
\inf_{\left(\binom{x}{y}, \binom{\bar x}{\bar y}\right) \in \N^2\times\N^2} \; C \left(\binom{x}{y}, \binom{\bar x}{\bar y}\right)>-\infty.
\end{equation}
The value $ C (\binom{x}{y}, \binom{\bar{x}}{\bar{y}}) $ shall be interpreted as the maintenance cost to have a
population with $x$ females and $y$ males, followed by a population with $\bar{x}$ females and $\bar{y}$ males. In
particular, we are assuming that the maintenance cost takes into account only two successive population censuses. This
assumption is made for simplicity and can be justified by observing that this model captures the main features of the
general case, when the maintenance cost is a function of a finite number of consecutive
censused-population sizes (see section~\ref{Discussion}).

In a finite-resource environment, it is reasonable to assume that, uniformly and independently
of the initial population size, the cost to generate and maintain a new population tends to infinity as its
size increases. In mathematical terms, the latter hypothesis can be expressed as follows
\begin{equation}\label{FiniteRessources1}
\lim_{\bar x + \bar y \, \to \, + \infty} \; \inf_{\binom{x}{y} \in \N^2} \; C \left(\binom{x}{y}, \binom{\bar x}{\bar y}\right) = +\infty.
\end{equation}

On the other hand, it is also reasonable to suppose that the population maintenance cost is, in some sense, more
affected by the current population than by the former one. Roughly speaking, such an assumption means that, although
the cost for a small initial population generating a very numerous new one may be high, the maintenance of a numerous
population has a high cost by itself, independently of its previous size. Therefore, we shall assume that there exists
a constant $ \mathfrak K_C >0 $ such that the cost of having $\binom{\bar x}{\bar y}$ in any census does not vary more
than $ \mathfrak K_C $ as a function of the possible values for the former population, or more precisely, we assume
\begin{equation}\label{FiniteRessources2}
\mathfrak K_C :=\sup_{\binom{\bar x}{\bar y} \in \N^2}
\left[ \sup_{\binom{x}{y} \in \N^2} C \left(\binom{x}{y}, \binom{\bar x}{\bar y}\right) -
\inf_{\binom{x}{y} \in \N^2} C \left(\binom{x}{y}, \binom{\bar x}{\bar y}\right) \right] < +\infty.
\end{equation}

\section{Historically adapted populations}\label{ExistenceHistoricalyAdaptedPop}

The intuitive idea is that a historically adapted population should minimize the maintenance cost along the time.
Although, in most of the cases, there is no meaning in talking about a minimum cost for infinite histories,
the idea of histories minimizing the cost along the time will be useful. As a matter of fact, this heuristic
motivation will allow to highlight a central functional equation that will lead us to a rigorous definition of
historically adapted populations.

\subsection{Heuristic motivation}\label{Heuristic}

In order to explore heuristically a definition of historically adapted populations,
consider that the function $ C $ is nonnegative\footnote{Mathematically there is no loss of
generality with such an assumption, since $C$ is bounded from below.}.
Note that the total maintenance cost of a population history $\bar{\omega}=\binom{\bar{x}_i}{\bar{y}_i}_{i\in\N}$
is given by $\sum_{k\ge 1}C(\binom{\bar{x}_{k-1}}{\bar{y}_{k-1}},\binom{\bar{x}_k}{\bar{y}_k})$, which
may clearly diverge. Assume for now that there exists a history with finite
total maintenance cost (that is, for which the series converges). Thus, the
smallest total cost for some history beginning from a given initial population $\binom{x_0}{y_0}$
is just
\begin{equation}\label{TotalCost1}
u\binom{x_0}{y_0}:=\inf_{\left(\binom{x_1}{y_1},\binom{x_2}{y_2},\ldots\right)}\left[\sum_{k\ge 1}C\left(\binom{x_{k-1}}{y_{k-1}},\binom{x_k}{y_k}\right)\right].
\end{equation}

Since we are assuming that the total cost is finite for some history, then $u\binom{x_0}{y_0}\in\R$ for any $\binom{x_0}{y_0}\in\N^2$.
Moreover, as the cost function $C$ is supposed to be nonnegative, obviously $u \ge 0$ everywhere.
Now, note that the above equation can be rewritten as
\begin{eqnarray}\label{TotalCost2}
\displaystyle u\binom{x_0}{y_0}
&=&\displaystyle \inf_{\binom{x_1}{y_1}} \; \inf_{\left(\binom{x_2}{y_2},\binom{x_3}{y_3},\ldots\right)}\left[C\left(\binom{x_0}{y_0},\binom{x_1}{y_1}\right)+\sum_{k\ge 2}C\left(\binom{x_{k-1}}{y_{k-1}},\binom{x_k}{y_k}\right)\right] \nonumber \\
&=&\displaystyle \inf_{\binom{x_1}{y_1}} \left[C\left(\binom{x_0}{y_0},\binom{x_1}{y_1}\right)+\inf_{\left(\binom{x_2}{y_2},\binom{x_3}{y_3},\ldots\right)}\sum_{k\ge 2}C\left(\binom{x_{k-1}}{y_{k-1}},\binom{x_k}{y_k}\right)\right] \nonumber \\
&=&\displaystyle \inf_{\binom{x_1}{y_1}} \left[C\left(\binom{x_0}{y_0},\binom{x_1}{y_1}\right)+u\binom{x_1}{y_1}\right].
\end{eqnarray}
Hypothesis~\eqref{FiniteRessources1} implies that the infimum in the above equation is actually a minimum.
Therefore, for each $\binom{x_0}{y_0}\in\N^2$, there exists $\binom{x_1}{y_1}\in\N^2$ such that
\begin{equation}\label{TotalCost3}
 u\binom{x_0}{y_0}= C\left(\binom{x_0}{y_0},\binom{x_1}{y_1}\right)+u\binom{x_1}{y_1}.
\end{equation}

One has now a recursive way to construct an interesting history. Indeed, given
an initial population $\binom{x_0}{y_0}$, we find $\binom{x_1}{y_1}$ which satisfies~\eqref{TotalCost3} and,
inductively, from the population $\binom{x_{i-1}}{y_{i-1}}$ at the moment $ i - 1 $, we obtain a population $\binom{x_i}{y_i}$ at the
subsequent moment such that $u\binom{x_{i-1}}{y_{i-1}}= C(\binom{x_{i-1}}{y_{i-1}},\binom{x_i}{y_i})+u\binom{x_i}{y_i}$.
Since its total maintenance cost is equal to the smallest one we could expect for some history beginning from $\binom{x_0}{y_0}$,
the history $\omega=\binom{x_i}{y_i}_{i\in\N}$ constructed by the above procedure will be called an adapted history.
We remark that there is not necessarily uniqueness and there might exist infinitely many adapted histories starting
from a given initial population.

\subsection{Rigorous definition}\label{FormalProof}

The existence of histories with finite total maintenance cost is a very strong demand made for our heuristic definition
of historically adapted populations. Besides being a tremendous restriction for the model, such an assumption implies
counterintuitively that the maintenance cost of these populations vanishes as time goes by.
Anyway, it is not totally naive the observation that an adapted history should be one for which in some sense $C(\binom{x_{k-1}}{y_{k-1}},\binom{x_k}{y_k})$ goes to the infimum of the cost function $C$ as $k$ tends to $+\infty$.
More important, the previous heuristic discussion leads us to propose a general definition of adapted histories,
which extends the intuitive idea of global minimizing histories for situations where the notion of
finite total maintenance cost has no meaning.

It is straightforward that whenever a function $u:\N^2\to\R$ is bounded from below and verifies
an equation like \eqref{TotalCost2}, we can use it to construct adapted histories.
If the maintenance cost function $C$ satisfies hypotheses~\eqref{BoundedBelow}, \eqref{FiniteRessources1}
and~\eqref{FiniteRessources2}, then one can show that there exist a bounded function\footnote{Such
a function is a fixed point for a kind of Lax-Oleinik operator (see Definition~\ref{Lax-Oleinik},
Theorem~\ref{TheoFixedPoint} and Remark~\ref{TheoFixedPoint2} in appendix A).} $u:\N^2\to\R$ and a real constant
$\gamma$ (both depending on the function $C$) such that
\begin{equation}\label{FixedPoint}
 u\binom{x}{y}=\inf_{\binom{\bar x}{\bar y} \in \N^2} \left[C\left(\binom{x}{y},\binom{\bar x}{\bar y}\right)+u\binom{\bar x}{\bar y}\right]-\gamma,
 \quad \text{ for all } \; \binom{x}{y} \in \N^2.
\end{equation}

Since $C$ satisfies \eqref{FiniteRessources1} and $u$ is bounded, we can again deduce that for each $\binom{x_0}{y_0}$
there exists $\binom{x_1}{y_1}$ such that $ u\binom{x_0}{y_0}=C(\binom{x_0}{y_0},\binom{x_1}{y_1})+u\binom{x_1}{y_1}-\gamma $.
Hence, given any initial population $\binom{x_0}{y_0}$, we define in a recursive way adapted histories starting
from $\binom{x_0}{y_0}$ as we have made in section~\ref{Heuristic}.

\begin{defn}\label{AdaptedHistory}
Let $u:\N^2\to\R$ be a bounded function satisfying the functional
equation~\eqref{FixedPoint} for some constant $\gamma$. Then $\omega=\binom{x_i}{y_i}_{i\in\N}$ is said
to be an adapted history for the maintenance cost $C$ if
\begin{equation}\label{DefinicaoPopHistAdap}
u\binom{x_{i-1}}{y_{i-1}}=C\left(\binom{x_{i-1}}{y_{i-1}},\binom{x_i}{y_i}\right)+u\binom{x_i}{y_i}-\gamma,
\quad \text{ for all } \; i\ge 1.
\end{equation}
\end{defn}

We remark that in this context the quantity $u\binom{x_0}{y_0}$ is not necessarily given by the
expression~\eqref{TotalCost1} and then it cannot be interpreted as the smallest total cost we would expect for any history starting from $\binom{x_0}{y_0}$. Anyway, if $\omega=\binom{x_i}{y_i}_{i\in\N}$ is an adapted history,
then it is easy to see that its average maintenance cost tends to $\gamma$ as time goes by, or in mathematical terms
\begin{equation}\label{MeanMaintenanceCost}
 \lim_{n\to\infty}\frac{1}{n}\sum_{k=1}^nC\left(\binom{x_{k-1}}{y_{k-1}},\binom{x_k}{y_k}\right)=\gamma.
\end{equation}
Therefore, it follows from the functional equation~\eqref{FixedPoint} that
$\gamma$ can be interpreted as the minimum asymptotic average maintenance cost we can expect for arbitrary
histories, and this minimum value $ \gamma $ is necessarily attained by any adapted history.
Thus, even without uniqueness of adapted histories, we have uniqueness of the quantity $\gamma$ (see Remark~\ref{TheoFixedPoint2} in appendix A).

If $\omega=\binom{x_i}{y_i}_{i\in\N}$ is an adapted history, we can recover the global minimizing property,
since any finite history $\omega{[m,n]}=\binom{x_i}{y_i}_{m\leq i\leq n}$ minimizes the maintenance cost among all
finite histories with the same initial and final population\footnote{In the language of solid state physics, such a
property means that the adapted histories behave like ground-states of one-dimensional crystal models (see, for
instance, \cite{ALeD}).}. More precisely, for all $ m,n\in\N $ with $ m < n $ and for any other population history
$ \bar{\omega} = \binom{\bar x_i}{\bar y_i}_{i\in\N} \in \Omega $ verifying $ \omega [m] = \bar{\omega} [m] $ and $ \omega [n] = \bar{\omega} [n] $, it follows
that
\begin{align*}
\sum_{k = m + 1}^n & C\left(\binom{x_{k-1}}{y_{k-1}},\binom{x_k}{y_k}\right)
=_{(1)} \sum_{k = m + 1}^n \left[ u \binom{x_{k - 1}}{y_{k - 1}} - u \binom{x_k}{y_k} +\gamma\right] \\
& =_{(2)} u \binom{x_m}{y_m} - u \binom{x_n}{y_n} +(n-m)\gamma
=_{(3)} \sum_{k = m + 1}^n \left[ u \binom{\bar{x}_{k - 1}}{\bar{y}_{k - 1}} - u \binom{\bar{x}_k}{\bar{y}_k} +\gamma \right] \\
& =_{(4)} \sum_{k = m + 1}^n \left[ \inf_{\binom{\bar x}{\bar y}}
\left[C\left(\binom{\bar x_{k-1}}{\bar y_{k-1}},\binom{\bar x}{\bar y}\right)+u\binom{\bar x}{\bar y}\right]- u \binom{\bar x_k}{\bar y_k} \right]
\le_{(5)} \sum_{k = m + 1}^n C\left(\binom{\bar x_{k-1}}{\bar y_{k-1}},\binom{\bar x_k}{\bar y_k}\right),
\end{align*}
where: $=_{(1)}$ is due to~\eqref{DefinicaoPopHistAdap}; $=_{(2)}$ comes from a telescopic series;
$=_{(3)}$ follows again from a telescopic series, as well as from the fact that $ \omega [m] = \bar{\omega} [m] $ and
$ \omega [n] = \bar{\omega} [n] $; $=_{(4)}$ is due to~\eqref{FixedPoint}; and $\le_{(5)}$ follows from
$ \inf_{\binom{\bar x}{\bar y}} [C(\binom{\bar x_{k-1}}{\bar y_{k-1}},\binom{\bar x}{\bar y})+u\binom{\bar x}{\bar y}] \le
C(\binom{\bar x_{k-1}}{\bar y_{k-1}},\binom{\bar x_k}{\bar y_k})+u\binom{\bar x_k}{\bar y_k}$.

A more pertinent point about historically adapted populations is that, as we will see in the next section,
they provide a theoretical argument in favor of the hypothesis of prevalence of stable sex
ratios for populations under stable environmental conditions.

\section{On the existence of the asymptotic average sex ratio}\label{AverageSexRatioSection}

Investigating the identification of a sex ratio, we will find out that when the population maintenance
cost takes into account the gender proportions, then a sex ratio will be observed in historically adapted
populations. In fact, since those populations are global minimizers of cost functions, one may argue that
a sex ratio will emerge as a consequence of the finiteness of available resources whenever gender densities
have a linear influence on the maintenance cost.

The gender proportions of the $ (i + 1)^{th} $ censused population $ \binom{x_i}{y_i} $ correspond
obviously to the quantities $ x_i/(x_i + y_i) $ and $ y_i/(x_i + y_i) $. Nevertheless, we have seen that there
exists at least one historically adapted population starting from any arbitrary population. Such a fact means that all first value for a sex ratio
can be then achieved and it shows that the analysis of initial data may be ineffective. Anyway, one can still ask whether a kind of sex-ratio equilibrium
will be reached in the latter generations. Mathematically, it corresponds to looking for some asymptotic sex ratio, that is, given an adapted history $\omega=\binom{x_i}{y_i}_{i\in\N}$, to ask for the existence of the limits
\begin{equation*}\label{AsymptoticSexRatio}
\lim_{i\to\infty}\frac{x_i}{x_i+y_i} \qquad \text{and} \qquad \lim_{i\to\infty}\frac{y_i}{x_i +y_i}.
\end{equation*}

The historically-adapted-population approach does not guarantee that convergence,
but it will allow to assure the convergence in average of sex ratio, namely, the existence of the limits
\begin{equation*}\label{AverageSexRatio}
\lim_{n\to\infty}\frac{1}{n}\sum_{k=1}^n\frac{x_k}{x_k+y_k}
\qquad \text{and} \qquad
\lim_{n\to\infty}\frac{1}{n}\sum_{k=1}^n\frac{y_k}{x_k+y_k}.
\end{equation*}
More important, if the sex ratio converges in average, then there are infinitely many arbitrarily long periods of
time for which it remains as close as one wants to the average limit.
This mathematical property might therefore explain the documented stability of sex ratio in the nature.

\subsection{Linearly perturbed maintenance cost functions}

Let $ C: \N^2 \times \N^2 \to \R $ be a cost function verifying hypotheses~\eqref{BoundedBelow}, \eqref{FiniteRessources1}
and~\eqref{FiniteRessources2}. Given a vector $ A=(a_1,a_2)\in\R^2 $, the linearly perturbed maintenance cost function with weights $a_1$ and $a_2$ on the gender densities is the function $C_A:\N^2\times\N^2\to \R$ given by
\begin{align*}
C_A\left(\binom{x}{y}, \binom{0}{0}\right) & = C\left(\binom{x}{y}, \binom{0}{0}\right), \quad \text{and} \\
C_A\left(\binom{x}{y}, \binom{\bar x}{\bar y}\right) & = C\left(\binom{x}{y}, \binom{\bar x}{\bar y}\right) +
a_1\frac{\bar{x}}{\bar{x}+\bar{y}}+a_2\frac{\bar{y}}{\bar{x}+\bar{y}}, \quad \text{if } \; \bar x + \bar y > 0.
\end{align*}

Notice that $a_1$ and $a_2$ assign cost (or benefits if negative) on the latest gender densities of the population.
Besides, the original maintenance cost function $C$ could include linear and nonlinear feedbacks for the gender densities. We shall study the asymptotic average sex ratio for historically adapted populations
with respect to perturbed maintenance cost functions in the above form.

It is straightforward that $C_{(0,0)}=C$. Besides, $ C_A $ converges uniformly to $ C $ as the vector $A$ tends to $(0,0)$
(see \eqref{C-distance}). Much more crucial is the fact that hypotheses \eqref{BoundedBelow}, \eqref{FiniteRessources1} and
\eqref{FiniteRessources2} also hold for the perturbed cost $C_A$. Thus, we can apply the result of
section~\ref{FormalProof}
to get the existence of $\omega^A=\binom{x^A_i}{y^A_i}_{i\in\N} \in \Omega$, an adapted history for the maintenance cost $C_A$.
Therefore, for a fixed cost function $ C $, we can consider the map $ \Gamma_C : \R^2 \to \R $ given by
$$ \Gamma_C(A) =
\lim_{n\to\infty}\frac{1}{n}\sum_{k=1}^n C_A\left(\binom{x^A_{k-1}}{y^A_{k-1}},\binom{x^A_k}{y^A_k}\right),
\quad \text{ for all } \; A \in \R^2. $$
(The above function is well defined due to \eqref{MeanMaintenanceCost}.)

Notice that if $\gamma$ is the minimum asymptotic average maintenance cost with respect to $C$,
then clearly $\Gamma_C(0,0)=\gamma$. Furthermore, one can easily show that $\Gamma_C$ is a concave
application, which, in particular, means that $\Gamma_C$ is continuous everywhere and differentiable
almost everywhere with respect to the Lebesgue measure\footnote{See Proposition~\ref{Concave} and Remark~\ref{Concave2}
in appendix A.}.
As a matter of fact, one may say a little more on the differentiable behavior of the map $ \Gamma_C $, since one can show that
\begin{equation}\label{Gama e f}
\Gamma_C(a_1, a_2) = f_C(a_1 - a_2) + \gamma + a_2, \quad \forall \; (a_1, a_2) \in \R^2,
\end{equation}
where $ f_C : \R \to \R $ is a concave function such that $ f_C (0) = 0 $ (see Remark~\ref{Concave2}). Thus, for almost all $ \Delta \in \R $, the map
$ \Gamma_C $ is actually differentiable along the straight line $ \{(a, a - \Delta) : a \in \R\} $.

The points of differentiability of $ \Gamma_C $ are essential for the discussion on the existence of an asymptotic average sex ratio
for historically adapted populations. Let then $\nabla \Gamma_C(A)$ denote the gradient vector of the function $\Gamma_C$ at
the point $A \in \R^2$. We are able to show that, whenever $ A = (a_1, a_2) $ is a point of differentiability of $ \Gamma_C $, one necessarily has
\begin{equation}\label{AverageSexRatio2}
\lim_{n \to \infty} \frac{1}{n}\sum_{k = 1}^n\left(\frac{x^A_k}{x^A_k + y^A_k}, \frac{y^A_k}{x^A_k + y^A_k}\right) = \nabla \Gamma_C(A) =
\big(f_C'(a_1 - a_2), 1 - f_C'(a_1 - a_2)\big),
\end{equation}
for any adapted history $\omega^A=\binom{x^A_i}{y^A_i}_{i\in\N} \in \Omega$ with respect to the maintenance cost function $ C_A $.
This is precisely the statement of Theorem~\ref{AverageSexRatioTheorem} in appendix A.

First of all, \eqref{AverageSexRatio2} means that, for almost all linearly perturbed maintenance cost function, the respective
historically adapted populations do share the same asymptotic average sex ratio. Moreover, this common value is constant for
each family of cost functions $ \{C_{(a, a - \Delta)}\}_{a \in \R} $, when $ \Delta $ is a point of differentiability of $ f_C $.
In particular, since $f_C$ is concave, then $f_C'$ is non-increasing whenever is defined. Therefore, if $ \Delta_1 < \Delta_2 $
are two points of differentiability of $ f_C $, then the asymptotic average female proportion in historically adapted populations
with respect to $ C_{(a, a - \Delta_1)} $ will not be less than the one in historically adapted populations with respect to
$ C_{(b, b - \Delta_2)} $.

The convergence in average of the sex ratio has a main consequence: a kind of stability along time for this biological feature.
In fact, a simple lemma (see Lemma~\ref{LemaAnalise}) assures that there will exist infinitely many arbitrarily long periods of time
for which sex ratios must be as close as desired to the average limit. In more mathematical terms, suppose $ A = (a_1, a_2) \in \R^2 $ is a
point of differentiability of the map $ \Gamma_C $, and $\omega^A=\binom{x^A_i}{y^A_i}_{i\in\N} \in \Omega$ is an adapted history with respect
to the maintenance cost function $ C_A $. Hence, given $ \epsilon > 0 $ arbitrarily small and $ M > 0 $ as large as one wants, there are infinitely
many finite histories $ \omega^A [m, n] = \binom{x^A_i}{y^A_i}_{m \le i \le n} $ such that $ n - m \ge M $ and
$$ \left | \sum_{k = m + 1}^n \left( \frac{x_k^A}{x_k^A + y_k^A} - f_C'(a_1 - a_2) \right) \right | < \epsilon. $$
This property might clearly provide a reasonable explanation for the observed stability of sex ratio in nature, underlying the major
role of the non-trivial interaction between a finite-resource environment and the growth and maintenance of a two-sex population.

\subsection{From perturbed to non-perturbed maintenance costs}

We have an asymptotic average sex ratio for historically adapted populations with respect to almost all
linearly perturbed maintenance cost function. However, $\Gamma_C$ may be non-differentiable along
countably many straight lines $ \{(b, b - \Delta) : b \in \R\}$, and the limit~\eqref{AverageSexRatio2}
is only guaranteed if $ A \in \R^2 $ is a point of differentiability of $\Gamma_C$. Note that this limit
could exist for some point of non-differentiability of $ \Gamma_C $, but
the argumentation in the proof of Theorem~\ref{AverageSexRatioTheorem} cannot be used
to decide if this is the case. In particular, one cannot assure that $\Gamma_C$ is differentiable at $ A=(0,0) $,
which would imply the existence of an asymptotic average sex ratio for
historically adapted populations with respect to the original non-perturbed maintenance cost function $ C $.

Although we cannot always guarantee the existence of an asymptotic average sex ratio, the characterization of
$\Gamma_C$ given in~\eqref{Gama e f} allows us to deduce that even in the worst case the average sex ratio takes
values in some fixed interval (see Theorem~\ref{NonDifferentiableCase}).
As a matter of fact, for all $ B = (b, b - \Delta) \in\R^2 $,
there exist real constants $ L_{\Delta} $ and $ R_{\Delta} $ (depending only on the cost function $ C $ and on the real parameter $ \Delta $) such that
\begin{equation*}
0 \le L_{\Delta} \le \liminf_{n \to \infty} \frac{1}{n}\sum_{k = 1}^n\frac{x^B_k}{x^B_k + y^B_k}
\le \limsup_{n \to \infty} \frac{1}{n}\sum_{k = 1}^n\frac{x^B_k}{x^B_k + y^B_k}\le R_{\Delta} \le 1,
\end{equation*}
whenever $ \omega^B=\binom{x^B_i}{y^B_i}_{i\in\N} \in \Omega $ is an adapted history for the maintenance cost function $ C_B $.

For any points $ B=(b,b - \Delta) $ and $ \bar B=(\beta, \beta - \bar \Delta) $ with
$ \Delta < \bar \Delta $, one can show that $R_{\bar \Delta} \le L_{\Delta}$ (see
Remark~\ref{NonDifferentiableCaseRemark2}). Hence the respective intervals $[L_{\Delta},\ R_{\Delta}]$ and
$[L_{\bar \Delta},\ R_{\bar \Delta}]$ may intersect each other only at their common boundary.
In particular, a small perturbation, let us say, on $ \bar \Delta $ will imply that the
corresponding average sex ratios of historically adapted populations must take its values outside the interval
$(L_{\bar \Delta},\ R_{\bar \Delta})$. One might interpret this fact
as a kind of instability of average sex ratios for points of non-differentiability of $ \Gamma_C $.

To illustrate the above discussion, let us consider an extreme situation. Suppose that there exists
some point $\bar B=(\beta,\beta - \bar \Delta)$ such that $L_{\bar\Delta}=0$ and $R_{\bar\Delta}=1$.
It is straightforward that, for any other point $B = (b, b - \Delta) $ with $ \Delta \ne \bar \Delta $,
we have
$$\text{either}\qquad\lim_{n \to \infty} \frac{1}{n}\sum_{k = 1}^n\frac{x^B_k}{x^B_k + y^B_k}=0\qquad\text{or}
\qquad \lim_{n \to \infty} \frac{1}{n}\sum_{k = 1}^n\frac{x^B_k}{x^B_k + y^B_k}=1,$$
for all adapted history $ \omega^B=\binom{x^B_i}{y^B_i}_{i\in\N} \in \Omega $ with respect to $ C_B $.
So only for cost functions $ C_{(\bar b, \bar b - \bar \Delta)} $ it could exist
historically adapted populations with two genders coexisting as time goes by.

A brief concluding remark is that, for the special case of the non-perturbed maintenance cost function $ C=C_{(0,0)}$,
there always exist constants $ 0 \le L_0 \le R_0 \le 1 $ such that average sex ratios of historically
adapted populations for $ C $ either converge to some point of the interval $ [L_0,\ R_0] $ or take values in this interval in a periodic or random way, without convergence.

\section{Discussion}\label{Discussion}

We proposed a new theoretical paradigm for sex-ratio problems: reproductive interactions are supposed to have
interconnectedness governed by a maintenance cost function depending explicitly on the size of male and female populations.
By considering an environment with finite resources, we are compelled to take three hypotheses on the maintenance
cost function: there exists a minimum cost (or a maximum benefit) which could be achieved by some population; the cost
diverges to infinity as the latest population increases; the cost is dominated by the current population size.
In this framework, we were able to show that there always exist historically adapted populations, which are
populations minimizing the maintenance cost along time. Furthermore, the main result established here has guaranteed
that, for almost all linearly perturbed maintenance cost function, the average sex ratios of the respective historically
adapted populations do converge.

We emphasize that the proposed model has mainly a heuristic focus and does not try to explain mechanisms which could lead some
population to be historically adapted. One obviously recognizes the fundamental importance of biological researches into a
possible validation of such an approach. The main point seems to be an affirmative answer to the challenge of representing
interactions of an actual biological population through a maintenance cost function with the desired features.

We notice that formalism developed in previous sections and its consequences can be immediately generalized to various other situations.
We would like to briefly discuss some of them.

\paragraph{Cost dependence on a finite number of consecutive censuses.}
One may consider a maintenance cost function $ C : \N^{2L} \to \R $ depending on
$ L \ge 2 $ consecutive population census, which can be seen again depending
on two coordinates $ C : \N^{2(L - 1)} \times \N^{2(L - 1)} \to \R $ via the identification
\begin{center}
$ C \left( \binom{x_1}{y_1}, \binom{x_2}{y_2}, \ldots, \binom{x_L}{y_L} \right) =
C \Big( \big( \binom{x_1}{y_1} \;\; \binom{x_2}{y_2} \;\; \ldots \;\; \binom{x_{L-1}}{y_{L-1}} \big)^T \hspace{-.1cm}, \, \big( \binom{x_2}{y_2} \;\; \binom{x_3}{y_3} \;\; \ldots \;\; \binom{x_L}{y_L} \big)^T \Big). $
\end{center}
One may now use such a point of view to rewrite the assumptions on the cost function and to easily obtain the analogous results for historically adapted populations.

\paragraph{Age-structured population models.}
We can introduce, for instance, the quantities of newborns of each gender.
Hence, if newborns are included as a cost factor, then, for almost all perturbed cost function, there shall
exist an identifiable primary sex ratio in historically adapted populations. More generally, one may analyze
a maintenance cost function depending on $ M \ge 2 $ age classes for both genders, namely, a function
$ C : \N^{2M} \times \N^{2M} \to \R $,
$ C\big(\binom{(x_1, x_2, \ldots, x_M)}{(y_1, y_2, \ldots, y_M)}, \binom{(\bar x_1, \bar x_2, \ldots, \bar x_M)}{(\bar y_1, \bar y_2, \ldots, \bar y_M)}\big) $.

\paragraph{Sequential and simultaneous hermaphroditism.} By adding variables in our model, we can without difficulty extend our study
to the occurrence at the same time of separate and combined sexes in some biological system. For instance, if $ X $
and $ Y $ denote the sizes of the dioecious part of the population, concerning the sex reversal part, let $ h_x $
and $ H_y $ be the number of sequential hermaphrodites reproducing early in life, respectively, as females and as males.
The number of individuals after sex changes will be indicated then by $ h^y $ and $ H^x $, respectively. At last,
let $ Z $ denote the number of individuals having simultaneously both male and female reproductive organs.
Therefore, the distribution of dioecy versus hermaphroditism can be investigated, for example, by the means of a maintenance cost function
$ C : \N^7 \times \N^7 \to \R $,
\begin{center}
$ C \Big( \big( X \;\; Y \;\; h_x \;\; h^y \;\; H_y \;\; H^x \;\; Z \big)^T \hspace{-.1cm}, \,
\big( \bar X \;\; \bar Y \;\; \overline{h_x} \;\; \overline{h^y} \;\; \overline{H_y} \;\; \overline{H^x} \;\; \bar Z
\big)^T \Big). $
\end{center}

\paragraph{Periodic cost functions.} The population maintenance cost may vary periodically along time.
Such a situation corresponds to consider a family of cost functions $ C_1, C_2, \ldots, C_N : \N^2 \times \N^2 \to \R $
and, for any finite history  $\omega{[m,n]}=\binom{x_i}{y_i}_{m\leq i\leq n}$, a total cost
$ \sum_{k = m + 1}^n C_{k - 1 \; (\text{mod } N)} ( \binom{x_{k - 1}}{y_{k - 1}}, \binom{x_k}{y_k} ) $.
The analysis of the periodic case may be reduced to our time-independent case just by introducing a
conjunction cost map
\begin{center}
$ C \left(\binom{x}{y}, \binom{\bar x}{\bar y} \right) :=
\inf_{\binom{x_1}{y_1}, \binom{x_2}{y_2}, \ldots, \binom{x_{N - 1}}{x_{N - 1}} \in \N^2}
\left[C_1\left(\binom{x}{y}, \binom{x_1}{y_1}\right) + C_2\left(\binom{x_1}{y_1}, \binom{x_2}{y_2}\right) + \ldots +
C_N\left(\binom{x_{N-1}}{y_{N-1}}, \binom{\bar x}{\bar y}\right)\right]. $
\end{center}

It would certainly be very interesting to take into account two or more generalizations at the same application. For instance, an age-structured model with
sex reversal individuals might help to understand whether there should exist a special age for sex change. Here again, from the biological perspective,
in this particular situation as well as in many other examples of potential applications, such a form of modeling requires first to describe more explicitly properties of a maintenance cost function regulating the reproductive
interactions of a given population.

\section*{Appendix A}\label{Appendix}

In this appendix we shall give the mathematical proofs of the results previously discussed.
From now on, let $ C : \N^2 \times \N^2 \to \R $ be a function verifying assumptions
\eqref{BoundedBelow}, \eqref{FiniteRessources1} and \eqref{FiniteRessources2}. The main idea is to associate to such a maintenance cost
function a kind of Lax-Oleinik operator and use its fixed points to construct historically adapted populations as well as to study their asymptotic
properties. Lax-Oleinik fixed point techniques have been successfully explored in several areas. A very important example comes from calculus of
variations: the Lax-Oleinik semigroup, which is essential in Fathi's weak KAM theory for Lagrangian mechanics (see \cite{Fathi}).

First of all, we need to introduce the spaces on which our Lax-Oleinik operator will act.
Denote by $\ell^\infty(\N^2)$ the set of all real valued bounded functions on $\N^2$, and
denote by $\ell^\infty(\N^2) / \R $ the set of all real valued bounded functions on $\N^2$ modulo constants, that is,
the set of equivalence classes $ [f]:=\{g\in\ell^\infty(\N^2):\ f-g\equiv cte\}$.
Both $\ell^\infty(\N^2)$ and $\ell^\infty(\N^2) / \R $ are Banach spaces with norms
$ \| f \|_\infty:=\sup_{\binom{x}{y}\in\N^2} |f\binom{x}{y}|$ and $ \| [f] \|_{\#} := \inf_{\kappa \in \R} \| f + \kappa \|_{\infty} $, respectively.

\begin{defn}[The Lax-Oleinik operator] \label{Lax-Oleinik}
Let $ T_C $ be the operator acting on $ \ell^\infty(\N^2) $ by
\begin{equation*}\label{Operator}
T_C f \binom{x}{y}:=\inf_{\binom{\bar x}{\bar y}\in\N^2}\left[C\left(\binom{x}{y},\binom{\bar x}{\bar y}\right)+f\binom{\bar x}{\bar y}\right],
\quad \forall \; \binom{x}{y}\in\N^2,
\end{equation*}
whenever $ f $ is a real valued bounded function on $ \N^2 $.
\end{defn}

Notice that the operator $ T_C $ is well defined, since
\begin{equation*}\label{BoundOperator}
\inf C - \| f \|_\infty \le T_C f \le \sup_{\binom{x}{y}\in\N^2} C\left(\binom{x}{y},\binom{0}{0}\right) + \| f \|_\infty \le
\mathfrak K_c + C\left(\binom{0}{0},\binom{0}{0}\right) + \| f \|_\infty.
\end{equation*}
The above upper bound also implies that the infimum in the definition of the Lax-Oleinik operator is actually a minimum.
Indeed, as the cost function $ C $ verifies hypothesis~\eqref{FiniteRessources1}, $ T_C f \binom{x}{y} $ will be
selected among a finite number of values. Furthermore, note that $ T_C (f + \kappa) = T_C (f) + \kappa $ for any $ \kappa \in \R $. Thus, we can consider $T_C$ acting on $ \ell^\infty(\N^2) / \R $.

\begin{theo}\label{TheoFixedPoint}
The operator $ T_C : \ell^\infty(\N^2) / \R \to \ell^\infty(\N^2) / \R $ has a fixed point.
\end{theo}

\begin{proof}
We remark first that
$$ 2 \big\| [f] \big\|_{\#} = \text{osc}(f) := \sup_{\binom{x}{y}, \binom{\bar x}{\bar y}\in\N^2} \left[f\binom{x}{y}-f\binom{\bar x}{\bar y}\right],$$
where $f$ is any element of the equivalence class $[f]$.

Given $ \lambda \in (0, 1) $, let $ M_\lambda $ be the multiplication by $ 1 - \lambda $ acting on $ \ell^\infty(\N^2) / \R $.
The operator $ T_C \circ M_\lambda $ is a contraction and has therefore a fixed point $ [u_\lambda] \in \ell^\infty(\N^2) / \R $,
that is, $ (T_C \circ M_\lambda)[u_\lambda]=T_C[(1-\lambda)u_\lambda]=[u_\lambda]$.
Hence, observe that
$$ \big\| [u_\lambda] \big\|_{\#} = \frac{1}{2}\text{osc}\left((T_C \circ M_\lambda)u_\lambda\right) =
\frac{1}{2}\text{osc}\left(T_C (1 - \lambda)u_\lambda\right) \le \frac{\mathfrak K_C}{2}. $$
In particular, the family $ \{[u_\lambda]\}_{\lambda \in (0, 1)} $ has an accumulation point $[u] \in \ell^\infty(\N^2) / \R  $ as $ \lambda $ goes to zero.
Choose $\lambda_i\to 0$ such that $[u_{\lambda_i}]\to [u]$ as $i\to \infty$. Since $ T_C $ is 1-Lipschitz, we have
$$T_C [u] = \lim_{i\to+\infty}T_C [(1 - \lambda_i)u_{\lambda_i}] = \lim_{i\to+\infty}[u_{\lambda_i}]=[u]. $$
\end{proof}

\begin{remark}\label{TheoFixedPoint2}
We have $ T_C [u] = [u] $ for some equivalence class $ [u] \in \ell^\infty(\N^2) / \R $.
Therefore, if $u\in \ell^\infty(\N^2)$ is an element of the equivalence class $[u]$, it follows that
\begin{equation*}
 u\binom{x}{y} + \gamma = T_C u\binom{x}{y} =
\min_{\binom{\bar x}{\bar y}\in\N^2}\left[C\left(\binom{x}{y},\binom{\bar x}{\bar y}\right)+u\binom{\bar x}{\bar y}\right],
\quad \forall \; \binom{x}{y}\in\N^2,
\end{equation*}
for some real constant $ \gamma $. Recall that such a functional equation allows to construct an adapted history
$\omega=\binom{x_i}{y_i}_{i\in\N}$ starting with any given initial population $\binom{x_0}{y_0}$. In particular, it is easy to see that,
for any adapted history $\omega=\binom{x_i}{y_i}_{i\in\N}$ and all arbitrary history $\bar \omega = \binom{\bar x_i}{\bar y_i}_{i\in\N}$,
\begin{multline*}
\lim_{n \to \infty} \frac{1}{n}\sum_{k=1}^n \left[C\left(\binom{x_{k-1}}{y_{k-1}},\binom{x_k}{y_k}\right) + u\binom{x_k}{y_k} - u\binom{x_{k-1}}{y_{k-1}}\right]
= \\ = \gamma \le
\liminf_{n \to \infty} \frac{1}{n} \sum_{k=1}^n \left[C\left(\binom{\bar x_{k-1}}{\bar y_{k-1}},\binom{\bar x_k}{\bar y_k}\right)+ u\binom{\bar x_k}{\bar y_k} - u\binom{\bar x_{k-1}}{\bar y_{k-1}}\right].
\end{multline*}
Thus, one clearly has~\eqref{MeanMaintenanceCost} and
\begin{equation*}
\gamma = \inf_{\binom{\bar x_i}{\bar y_i}_{i\in\N} \in \, \Omega} \,
\liminf_{n \to \infty} \frac{1}{n} \sum_{k=1}^n C\left(\binom{\bar x_{k-1}}{\bar y_{k-1}},\binom{\bar x_k}{\bar y_k}\right).
\end{equation*}
\end{remark}

Now, given $ A \in \R^2 $, consider the perturbed maintenance cost

\begin{equation*}
C_A\left(\binom{x}{y}, \binom{\bar x}{\bar y}\right):=C\left(\binom{x}{y}, \binom{\bar x}{\bar y}\right)+\left\langle A,\left(\frac{\bar{x}}{\bar{x}+\bar{y}},\frac{\bar{y}}{\bar{x}+\bar{y}}\right)\right\rangle,
\end{equation*}
with the convention that zero over zero is equal to zero.
It is straightforward that $ C_A $ verifies~\eqref{BoundedBelow}, \eqref{FiniteRessources1} and~\eqref{FiniteRessources2}, and $ \mathfrak K_{C_A} = \mathfrak K_C $. Furthermore
\begin{equation}\label{C-distance}
\left\| C_A - C_B \right\|_\infty =
\sup_{\binom{\bar x}{\bar y}\in\N^2}
\left|\left\langle A-B,\left(\frac{\bar{x}}{\bar{x}+\bar{y}},\frac{\bar{y}}{\bar{x}+\bar{y}}\right)\right\rangle\right|
\le \| A-B \|, \quad \forall \, A, B \in \R^2.
\end{equation}

Let $\Gamma_C:\R^2\to\R$ be the map defined by
$$ \Gamma_C(A) := \inf_{\binom{\bar x_i}{\bar y_i}_{i\in\N} \in \, \Omega} \,
\liminf_{n \to \infty} \frac{1}{n} \sum_{k=1}^n C_A\left(\binom{\bar x_{k-1}}{\bar y_{k-1}},\binom{\bar x_k}{\bar y_k}\right). $$
Obviously $ \Gamma_C(0,0) = \gamma $. Moreover, we have that

\begin{prop}\label{Concave}
The function $\Gamma_C$ is concave.
\end{prop}

\begin{proof}
Given $ A, B \in \R^2 $ and $ t \in [0, 1] $, let $\omega=\binom{x_i}{y_i}_{i\in\N}$
be an adapted history with respect to the cost function $C_{tA+(1-t)B}$.
Therefore, by the very definition of $ \Gamma_C $, we get
\begin{eqnarray*}
t\Gamma_C(A)+(1-t)\Gamma_C(B)
& \le & \lim_{n \to \infty} \frac{1}{n} \sum_{k = 1}^n
\left[ tC_A\left(\binom{x_{k-1}}{y_{k-1}},\binom{x_k}{y_k}\right)+(1-t)C_B\left(\binom{x_{k-1}}{y_{k-1}},\binom{x_k}{y_k}\right)\right] \\
&=&
\lim_{n \to \infty}\frac{1}{n}\sum_{k = 1}^n C_{tA + (1- t)B}\left(\binom{x_{k-1}}{y_{k-1}},\binom{x_k}{y_k}\right)
= \Gamma_C(tA + (1-t)B).
\end{eqnarray*}
\end{proof}

\begin{remark}\label{Concave2}
A real valued concave function on $\R^n$ is locally Lipschitz continuous\footnote{Actually, it is not hard to directly check that,
for all $ A, B \in \R^2$, we have $ | \Gamma_C(A) - \Gamma_C(B) | \le \| A-B \| $.} and hence, by Rademacher's theorem, differentiable almost everywhere with respect to the Lebesgue measure. Thus, Theorem \ref{Concave} implies that Lebesgue-almost every $ A \in \R^2 $ is a point of differentiability
of the map $ \Gamma_C $.
As a matter of fact, one may be a little more precise on the description of the points of differentiability of $ \Gamma_C $.
To that end, notice that we can write
$$ \Gamma_C(a_1, a_2) = f_C(a_1 - a_2) + \gamma + a_2, \quad \forall \, (a_1, a_2) \in \R^2, $$
with $ f_C : \R \to \R $ defined by
$$ f_C (\Delta) := \inf_{\binom{\bar x_i}{\bar y_i}_{i\in\N} \in \, \Omega} \, \liminf_{n\to\infty}
\frac{1}{n}\sum_{k=1}^n\left[C\left(\binom{\bar x_{k-1}}{\bar y_{k-1}},\binom{\bar x_k}{\bar y_k}\right)+\frac{\bar x_k}{\bar x_k+\bar y_k}\Delta-\gamma\right],
\quad \forall \, \Delta\in\R.$$
Certainly $ f_C(0) = 0 $. Moreover, as in the proof of Proposition~\ref{Concave}, one may immediately verify that the function $ f_C $ is concave and
therefore differentiable almost everywhere with respect to the Lebesgue measure on the real line. So we conclude that, for Lebesgue-almost every $ \Delta\in\R $,
the map $ \Gamma_C $ is indeed differentiable along the straight line $ \{(a, a - \Delta) : a \in \R\} $.
\end{remark}

The next theorem shows that points of differentiability of $ \Gamma_C $ play a crucial role on the study of average sex ratio for historically adapted
populations. Its proof is very similar to Gomes' argument for the asymptotic behavior of optimal trajectories defined by discrete viscosity solutions
(see \cite{Gomes}).

\begin{theo}\label{AverageSexRatioTheorem}
Let $ A = (a_1, a_2) \in \R^2 $ be a point of differentiability of $\Gamma_C$, and let $\omega^A=\binom{x^A_i}{y^A_i}_{i\in\N}$ be an adapted history
for the maintenance cost $C_A$. Then, one has
\begin{equation*}
\lim_{n \to \infty} \frac{1}{n}\sum_{k = 1}^n\left(\frac{x^A_k}{x^A_k + y^A_k}, \frac{y^A_k}{x^A_k + y^A_k}\right) =
\nabla \Gamma_C(A) = \big(f_C'(a_1 - a_2), 1 - f_C'(a_1 - a_2)\big).
\end{equation*}
\end{theo}

\begin{proof}
Let $ u_A \in \ell^\infty(\N^2) $ be such that $ u_A\binom{x^A_i}{y^A_i} + \Gamma_C(A) = T_{C_A} u_A\binom{x^A_i}{y^A_i} =
C_A(\binom{x^A_i}{y^A_i}, \binom{x^A_{i + 1}}{y^A_{i + 1}}) + u_A\binom{x^A_{i + 1}}{y^A_{i + 1}} $, for all $ i\in\N $.
Therefore, for all $n\geq 1$, we obtain
$$ u_A \binom{x^A_0}{y^A_0} = \sum_{k = 1}^n C_A\left(\binom{x^A_{k-1}}{y^A_{k-1}},\binom{x^A_k}{y^A_k}\right) + u_A \binom{x^A_n}{y^A_n} - n \Gamma_C(A). $$
For $ h > 0 $ and $ B \in \R^2 $, let $ u_{A + hB} \in \ell^\infty(\N^2) $ be such that
$ T_{C_{A + hB}} u_{A + hB} = u_{A + hB} + \Gamma_C(A + hB) $. It is straightforward that
$$ u_{A+hB} \binom{x^A_0}{y^A_0} \le
\sum_{k = 1}^n C_{A+hB}\left(\binom{x^A_{k-1}}{y^A_{k-1}},\binom{x^A_k}{y^A_k}\right) + u_{A+hB} \binom{x^A_n}{y^A_n} - n \Gamma_C(A+hB),
\quad \forall \, n \ge 1. $$

Thus, clearly
\begin{multline*}
u_{A+hB}\binom{x^A_0}{y^A_0} - u_A\binom{x^A_0}{y^A_0} \le \\
\le h \sum_{k = 1}^n\left\langle B, \left(\frac{x^A_k}{x^A_k + y^A_k}, \frac{y^A_k}{x^A_k + y^A_k}\right)\right\rangle
- n \big(\Gamma_C(A+hB)-\Gamma_C(A)\big)+u_{A+hB}\binom{x^A_n}{y^A_n} - u_A\binom{x^A_n}{y^A_n}.
\end{multline*}
Since $ u_{A+hB}\binom{x^A_0}{y^A_0}-u_{A+hB}\binom{x^A_n}{y^A_n}+u_A\binom{x^A_n}{y^A_n}-u_A\binom{x^A_0}{y^A_0}\ge
-\text{osc}(T_{C_{A + hB}}u_{A+hB})-\text{osc}(T_{C_A}u_A)\ge-2\mathfrak K_C$, it follows that
$$ - \frac{2\mathfrak K_C}{hn} + \frac{\Gamma_C(A+hB)-\Gamma_C(A)}{h} \le
\left\langle B, \frac{1}{n} \sum_{k = 1}^n \left(\frac{x^A_k}{x^A_k + y^A_k}, \frac{y^A_k}{x^A_k + y^A_k}\right)\right\rangle. $$

The same argument can be applied to $ - B $ and hence we also deduce that
$$ \left\langle B, \frac{1}{n} \sum_{k = 1}^n \left(\frac{x^A_k}{x^A_k + y^A_k}, \frac{y^A_k}{x^A_k + y^A_k}\right)\right\rangle \le
\frac{2\mathfrak K_C}{hn} - \frac{\Gamma_C(A-hB)-\Gamma_C(A)}{h}. $$

So setting $ h=m/n $ for a fixed $m>0$ and taking $ n \to \infty $, as $A$ is a point of differentiability of $\Gamma_C$, from the last two inequalities
we get that
\begin{multline*}
- \frac{2\mathfrak K_C}{m} + \langle B,\nabla\Gamma_C(A)\rangle  \le
\left\langle B,\liminf_{n \to \infty} \frac{1}{n} \sum_{k = 1}^n \left(\frac{x^A_k}{x^A_k + y^A_k}, \frac{y^A_k}{x^A_k + y^A_k}\right)\right\rangle \le \\
\le \left\langle B,\limsup_{n \to \infty} \frac{1}{n} \sum_{k = 1}^n \left(\frac{x^A_k}{x^A_k + y^A_k}, \frac{y^A_k}{x^A_k + y^A_k}\right)\right\rangle\le
\frac{2\mathfrak K_C}{m}+\langle B,\nabla\Gamma_C(A)\rangle.
\end{multline*}
Finally, taking  $ m \to +\infty $, we obtain
$$ \left\langle B,\lim_{n \to \infty} \frac{1}{n}\sum_{k = 1}^n\left(\frac{x^A_k}{x^A_k + y^A_k}, \frac{y^A_k}{x^A_k + y^A_k}\right)\right\rangle = \langle B,\nabla\Gamma_C(A)\rangle, $$
and then, since the equality holds for all $ B \in \R^2 $, we conclude that
$$ \lim_{n \to \infty} \frac{1}{n}\sum_{k = 1}^n\left(\frac{x^A_k}{x^A_k + y^A_k}, \frac{y^A_k}{x^A_k + y^A_k}\right) = \nabla\Gamma_C(A). $$
\end{proof}

Despite of being a weak form of convergence, in fact convergence in average underlines a recurrence property of the sequence. More precisely, we have the following result from real analysis.

\begin{lem}\label{LemaAnalise}
Let $ \{\alpha_i\} \subset \R $ be a sequence such that $ \lim_{n \to \infty} (1/n) \sum_{k=1}^{n} \alpha_k = \alpha \in \R $.
Let $ \I \subset \N $ be a subset of positive density, that is,
$$ \lim_{n \to \infty} \frac{\# \{k \in \I : 1 \le k \le n\}}{n} =: \beta > 0. $$
Then, for all $ \epsilon > 0 $ and for any integer $ L > 0 $, there exist $ m, n \in \I $, with $ n > m \ge L $, such that
$$ \left | \sum_{k = m + 1}^n \left( \alpha_k - \alpha \right) \right | < \epsilon. $$
\end{lem}

For the convenience of the reader, we give a short proof of this lemma.

\begin{proof}
Without loss of generality, we can assume that $ \alpha = 0 $.
Fix $ \rho \in (0, \epsilon \beta / 8) $.
There exists then a positive integer $ n_0 \in \I $ such that
$$  \# \{k \in \I : 1 \le k \le n \} \ge \frac{\beta n}{2}  \quad \text{ and }
\quad  \left | \sum_{k=1}^{n} \alpha_k \right| \le \rho n, \quad \, \forall \, n \ge n_0. $$
We may suppose that $ n_0 \ge L $.
Clearly $ \{ \sum_{k=1}^{n} \alpha_k : n_0 \le n \le n_1 \} \subset [-\rho n_1, \rho n_1] $.
Considering thus
$$ n_1 \in \mathbb N \quad \text{with} \quad
n_1 > \max \left\{ n_0, \; \frac{4}{\beta} \# \{k \in \I : 1 \le k \le n_0\} \right\}, $$
we assure that
$$ \# \{k \in \I : n_0 < k \le n_1\}  =  \# \{k \in \I : 1 \le k \le n_1\} - \#\{k \in \I : 1 \le k \le n_0\}
 >  \frac{\beta n_1}{2} - \frac{\beta n_1}{4} = \frac{\beta n_1}{4}. $$
By the pigeonhole principle, there must be $ m, n \in \I \cap \{n_0, n_0 + 1, \ldots, n_1\} $,
with $ n > m $, such that
\begin{multline*}
\left| \sum_{k=m + 1}^{n} \alpha_k \right| = \left| \sum_{k=1}^{m} \alpha_k - \sum_{k=1}^{n} \alpha_k \right|
\le \frac{2 \rho n_1}{\# \{k \in \I : n_0 \le k \le n_1\} - 1} = \\
= \frac{2 \rho n_1}{\# \{k \in \I : n_0 < k \le n_1\}} < \frac{2 \rho n_1}{\beta n_1/4} = \frac{8 \rho}{\beta} < \epsilon.
\end{multline*}
\end{proof}

Concerning then the stability of the average sex ratio for historically adapted populations, one obtains an immediate consequence, namely:

\begin{cor}
Let $ A = (a_1, a_2) \in \R^2 $ be a point of differentiability of $\Gamma_C$, and let $\omega^A=\binom{x^A_i}{y^A_i}_{i\in\N}$ be an adapted history
for the maintenance cost function $C_A$. Then, for all $ \epsilon > 0 $ and $ M > 0 $, there exist infinitely many finite histories
$\omega^A[m,n]=\binom{x^A_i}{y^A_i}_{m \le i \le n}$, with $ n - m \ge M $, such that
$$\left | \sum_{k = m + 1}^n \left( \frac{y_k^A}{x_k^A + y_k^A} - \big(1 - f_C'(a_1 - a_2)\big) \right) \right | =
\left | \sum_{k = m + 1}^n \left( \frac{x_k^A}{x_k^A + y_k^A} - f_C'(a_1 - a_2) \right) \right | < \epsilon. $$
\end{cor}

\begin{proof}
Just apply the previous lemma to $ \alpha_i = \frac{x^A_i}{x^A_i + y^A_i} $ and $ \I = \{ \lceil M \rceil, 2\lceil M \rceil, 3\lceil M \rceil, \ldots \} $,
where $ \lceil M \rceil $ denotes the smallest integer greater than or equal to $ M $.
\end{proof}

We recall that the one-sided derivatives of $f_C$ at a point $\Delta$ are given by
$$f_C'(\Delta+):=\lim_{H\to 0^+}\frac{f_C(\Delta+H)-f_C(\Delta)}{H}\qquad\text{and}\qquad
f_C'(\Delta-):=\lim_{H\to 0^-}\frac{f_C(\Delta+H)-f_C(\Delta)}{H}.$$
Since $f_C$ is a real valued concave function, its one-sided derivatives are defined everywhere.
The next theorem uses the one-sided derivatives of $f_C$ to find an estimate for the asymptotic average sex ratio even when it does not converge.

\begin{theo}\label{NonDifferentiableCase}
Given $ B=(b,b-\Delta)\in\R^2 $, define $ L_\Delta:=f_C'(\Delta+) $ and $ R_\Delta := f_C'(\Delta-) $. Then, for
all adapted history $ \omega^B=\binom{x^B_i}{y^B_i}_{i\in\N} \in \Omega $ with respect to $ C_B $, one has
$$ L_\Delta \le \liminf_{n \to \infty} \frac{1}{n}\sum_{k = 1}^n\frac{x^B_k}{x^B_k + y^B_k} \le
\limsup_{n \to \infty} \frac{1}{n}\sum_{k = 1}^n\frac{x^B_k}{x^B_k + y^B_k} \le R_\Delta. $$
\end{theo}

\begin{proof}
We will prove the result only for the point $B = (0,0)$, since the proof for any other point is analogous.
Let $ \omega=\binom{x_i}{y_i}_{i\in\N} \in \Omega $ be an adapted history with respect to the cost function $ C $.
For any point $A=(a,a-H)$ with $ H>0 $, we clearly have
$$ \frac{1}{n}\sum_{k=1}^n C_A\left(\binom{x_{k-1}}{y_{k-1}}, \binom{x_k}{y_k}\right)=
\frac{1}{n}\sum_{k=1}^n C\left(\binom{x_{k-1}}{y_{k-1}}, \binom{x_k}{y_k}\right) + \frac{H}{n}\sum_{k=1}^n \frac{x_k}{x_k+y_k} + a - H. $$
Since $\liminf_{n\to\infty}(1/n)\sum_{k=1}^n C_A\left(\binom{x_{k-1}}{y_{k-1}}, \binom{x_k}{y_k}\right)\ge\Gamma_C(A)$ and
$\lim_{n\to\infty}(1/n)\sum_{k=1}^n C\left(\binom{x_{k-1}}{y_{k-1}}, \binom{x_k}{y_k}\right) = \Gamma_C(0,0)=\gamma$, we obtain
$$ \Gamma_C(A) \le \gamma + H \liminf_{n\to\infty}\frac{1}{n}\sum_{k=1}^n \frac{x_k}{x_k+y_k}+ a - H. $$
Therefore, as $ f_C(H) = \Gamma_C(A) - \gamma - (a - H) $, we get that
$$ \frac{f_C(H)-f_C(0)}{H} = \frac{f_C(H)}{H} \le \liminf_{n\to\infty}\frac{1}{n}\sum_{k=1}^n \frac{x_k}{x_k+y_k}, $$
which yields
$$ f_C'(0+) \le \liminf_{n \to \infty} \frac{1}{n}\sum_{k = 1}^n\frac{x_k}{x_k + y_k}. $$
One obtains the inequality $\limsup_{n \to \infty} (1/n) \sum_{k = 1}^n x_k/(x_k + y_k) \le f_C'(0-)$ in a similar way,
using points $ A=(a,a-H)$ with $ H<0 $.
\end{proof}

\begin{remark} \label{NonDifferentiableCaseRemark2}
If $B$ is a point of differentiability of $\Gamma_C$, then the left-sided and right-sided derivatives coincide and we clearly recuperate the statement of Theorem~\ref{AverageSexRatioTheorem}.
Notice also that, since $f_C$ is concave, then both maps $ \Delta \mapsto f_C'(\Delta+) $ and $ \Delta \mapsto f_C'(\Delta-) $ are non-increasing functions and verify $ 0\le f_C'(\Delta+) \le f_C'(\Delta-) \le 1 $ for all
$ \Delta \in \R $. For any points $ B=(b,b - \Delta)$ and
$ \bar B=(\beta,\beta - \bar \Delta )$, with $ \Delta < \bar \Delta $, it follows that
$$ R_{\bar \Delta} = f_C'(\bar \Delta+) \le f_C'(\Delta-) = L_{\Delta}, $$
which implies that the respective intervals $[L_{\Delta},\ R_{\Delta}]$ and $[L_{\bar \Delta},\ R_{\bar \Delta}]$
may intersect each other only at their common boundary.
\end{remark}

\acknowledgement{The authors thank both mathematics departments of UNICAMP and UFSC for the hospitality during the preparation of this manuscript, and their graduate programs for the financial support. M. Sobottka was supported
by CNPq-Brazil grant 304457/2009-4 and FUNPESQUISA/UFSC 2009.0138.}

\footnotesize{

}

\end{document}